\documentclass{article}
\usepackage{graphicx,amssymb,amsmath,amsthm}
\usepackage{subcaption}
\usepackage{xcolor}
\usepackage{complexity}
\usepackage{hyperref}
\usepackage{todonotes}
\usepackage{bm}
\usepackage[top=0.8in, bottom=0.9in, left=0.9in, right=0.9in]{geometry}

\newcommand{\subfigimg}[3][,]{%
  \setbox1=\hbox{\includegraphics[#1]{#3}}
  \leavevmode\rlap{\usebox1}
  \rlap{\hspace*{10pt}\raisebox{\dimexpr\ht1-1\baselineskip}{\small{#2}}}
  \phantom{\usebox1}
}

\newcommand{\htop}{h_\text{\normalfont top}}
\newcommand{\vleft}{v_\text{\normalfont left}}
\newcommand{\hbot}{h_\text{\normalfont bot}}
\newcommand{\vright}{v_\text{\normalfont right}}
\newcommand{\norm}[1]{\left\lVert#1\right\rVert}

\newcommand{\hMax}{\ensuremath{{h_{\max}}}}
\newcommand{\vMax}{\ensuremath{{v_{\max}}}}

\newtheorem{lemma}{Lemma}
\newtheorem{theorem}{Theorem}
\newtheorem{obs}{Observation}
\newtheorem{definition}{Definition}

\date{}
\title{\vspace*{-2ex}
Multiple Watchman Routes in Staircase Polygons\footnote{A preliminary version of this paper will appear in Proceedings of the 37th Canadian Conference on Computational
Geometry (CCCG 2025). \\A. B., B.~J. N.  and C. S.  are supported by grant 2021-03810 (Illuminate: provably good algorithms for guarding problems) from the Swedish Research Council (Vetenskapsr\r{a}det).}}

\author{
    Anna Brötzner\thanks{Department of Computer Science and Media Technology, Malm\"o University, Sweden, \texttt{\{anna.brotzner,bengt.nilsson.TS\}@mau.se}}
	\and
	Bengt~J. {Nilsson}$^\dagger$
    \and
    Christiane Schmidt\thanks{Department of Science and Technology, Link\"oping University, Sweden, \texttt{christiane.schmidt@liu.se}}
    }
 
\index{Author, First}

\begin{document}
\thispagestyle{empty}
\maketitle

\begin{abstract}
We consider the watchman route problem for multiple watchmen in staircase polygons, which are rectilinear $x$- and $y$-monotone polygons. For two watchmen, we propose an algorithm to find an optimal solution that takes 
quadratic time, improving on the 
cubic time of a trivial solution. For $m \geq 3$ watchmen, we explain where this approach fails, and present an approximation algorithm for the min-max criterion with only an additive error. 
\end{abstract}

\section{Introduction}
The watchman route problem asks for a shortest route inside a polygon, such that every point in the polygon is visible to some point on the route. It was first introduced by Chin and Ntafos~\cite{Optimum_WR}, who showed that the problem is \NP-hard for polygons with holes, but may be solved efficiently for simple polygons. Given a starting point, an optimal route can be computed in $O(n^3)$ time~\cite{tj2017_WRP}, and finding a solution without a fixed starting point takes a linear factor longer~\cite{tan2001_WRP-no-starting-point}.

The Watchman Route Problem has also been considered for multiple watchmen (a problem introduced in~\cite{cnn1991_mWRP-hist}). For histograms, efficient algorithms have been proposed for minimizing the total route length (min-sum)~\cite{cnn1991_mWRP-hist} and the length of the longest route (min-max)~\cite{ns1992_mWRP-hist-min-max}. 
For two watchmen, Mitchell and Wynters~\cite{mw-wrmg-91} proved \NP-hardness for the min-max objective in simple polygons. Recently, Nilsson and Packer presented a polynomial-time 5.969-approximation algorithm for the same objective in simple polygons~\cite{nilssonpacker2024}. 

In this paper, we consider a quite restricted class of polygons, staircase polygons, that for two watchmen allows us to assign the responsibility for guarding any edge solely to one of the two watchmen (and for two watchmen, seeing the complete polygon boundary is sufficient to see the whole polygon, as shown in~\cite{nilssonpacker2024}). Additionally, we show that the two routes can be separated by a diagonal between two reflex vertices. This enables a polynomial-time algorithm to compute the optimal two watchman routes (for both the min-max and the min-sum objective). 
Despite staircase polygons being so restricted, some of the observations we make do not hold for three or more watchman routes. 
This indicates a discrepancy in the computational complexity between the watchman route problem for one or two watchmen and for multiple watchmen.
We therefore propose an approximation algorithm for $m$ watchmen for the min-max criterion with only additive error that depends on the polygon.

\section{Notation and Preliminaries}
A polygon is called \emph{rectilinear} if all its edges are parallel to the $x$- or the $y$-axis of a given coordinate system, and \emph{$x$-monotone} (\emph{$y$-monotone}) if every line that is orthogonal to the $x$-axis ($y$-axis) intersects the polygon in exactly one connected interval. 
A \emph{staircase polygon} is a rectilinear polygon that consists of two $x$-$y$-monotone chains that intersect only at their endpoints. 
The vertices along a chain then alternatingly have angles of 90 and 270 degrees; we refer to the former ones as \emph{convex vertices} and to the latter ones as \emph{reflex vertices}. 
We call the polygonal chain of boundary edges that lie above and below the interior the \emph{ceiling} and the \emph{floor} of $P$, respectively. 
We consider the watchman route problem for multiple watchmen in staircase polygons, where every route is a closed route. 

\paragraph{Multiple Watchman Route Problem (\texorpdfstring{$\bm{m}$}{m}-WRP).}
Given a polygon $P$, and a number of watchmen $m$, find a shortest set of $m$ closed routes, with respect to the min-sum or min-max criterion, such that every point in $P$ is seen from at least one of the routes. \\

We denote the length of a route $w$ by $\lVert w \rVert$. In the following, if we refer to a watchman $w$ we mean that $w$ is the route that the watchman walks. 
We refer to a solution of the $m$-WRP as a set of $m$ watchman routes in $P$. 
In the following, we consider the $m$-WRP for the min-sum and the min-max criterion. Any statement on optimal watchman routes holds for either objective, unless stated otherwise. 

Let $P$ be a staircase polygon that is not guardable with two point guards. 
As $P$ is $x$- and $y$-monotone, we make the following
claim.

\begin{obs} \label{obs:interval}
    A watchman walking along a route $w$ with leftmost coordinate $x_{\min}$ and rightmost coordinate $x_{\max}$ sees all points $p \in P$ with $x(p) \in [x_{\min}, x_{\max}]$. 
\end{obs}
The analogous statement holds for the lowest coordinate $y_{\min}$ and the uppermost coordinate $y_{\max}$ of watchman route $w$. Watchman $w$ thus sees the contiguous part of the ceiling between $y_{\min}$ and $x_{\max}$, and the contiguous part of the floor between $x_{\min}$ and $y_{\max}$, see Figure~\ref{fig:contiguous-boundary-seen}. 
\begin{figure}
    \centering
    \includegraphics[width=0.45\linewidth]{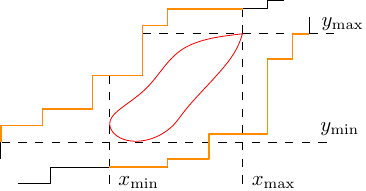}
    \caption{The orange parts of the polygon boundary are seen from the watchman route (red). }
    \label{fig:contiguous-boundary-seen}
\end{figure}

A {\em cut\/} is a directed line segment in $P$ with both end points on boundary of $P$ and where each interior point is an interior point of $P$. 
A cut always separates $P$ into exactly two sub-polygons of nonzero area.
If a cut is represented by the segment $[p,q]$ we say that the cut is directed from $p$ to $q$. For a cut $c$ in $P$, we define the {\em left polygon}, $L(c)$, to be the set of points in $P$ locally to the left of $c$ according to $c$'s direction.

Assume a counterclockwise walk of the boundary of $P$. Such a walk imposes a direction on each of the edges of $P$ in the direction of the walk. Consider a reflex vertex of $P$. The two edges incident to the vertex can each be extended inside $P$ until the extensions reach a boundary point. These extended segments form cuts given the same direction as the edge they are collinear to. We call these cuts {\em extensions}. 

Given a set of cuts, we say that a cut $c$ \emph{dominates} a cut $c'$ if the subpolygon $L(c)$ is a subset of $L(c')$. The non-dominated cuts are called \emph{essential cuts}. 
A cut $c$ is visited by a watchman if its route has a point on $c$ or in $L(c)$. 
The essential cuts are exactly those cuts that need to be visited if the polygon is to be seen by a single watchman. 
Clearly, visiting all essential cuts is also a necessary condition for a set of $m$ watchman~routes. 

A staircase polygon has at most four essential cuts: the leftmost vertical extension of the floor $\vleft$, the lowest horizontal extension of the ceiling $\hbot$, the rightmost vertical extension of the ceiling $\vright$, and the topmost horizontal extension of the floor $\htop$. 
Note that not necessarily all of these four extensions are essential cuts, but at least one of $\vleft$ and $\hbot$, and one of $\vright$ and $\htop$ is essential (if only one cut of a pair is essential, it dominates the other cut of the pair, and visiting the essential cut of such a pair guarantees that also the other cut of the pair is visited). 
For the sake of simplicity, we 
still
refer to all four of them as essential~cuts. 

For one watchman, an optimal solution is given by the shortest route that visits all essential extensions. An example is shown in Figure~\ref{fig:opt-solutions}(a). 
Chin and Ntafos~\cite{Optimum_WR}, prove that such a route can be computed in linear~time. 

\begin{theorem}\label{thm:linear-opt}
    (Theorem 2, Chin, Ntafos~\cite{Optimum_WR})
    A shortest watchman route in simple rectilinear polygons can be found in $O(n)$ time.
\end{theorem}

\begin{figure*}
\centering\hfill
\subfigimg[width=0.33\linewidth, page=1]{(a)}{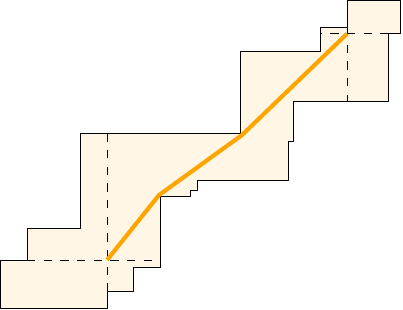}\hfill
\subfigimg[width=0.33\linewidth, page=2]{(b)}{2-watchmen-updated.pdf}\hfill
\subfigimg[width=0.33\linewidth, page=3]{(c)}{2-watchmen-updated.pdf}
\hfill\mbox{}
\caption{Optimal solutions for (a) one watchman, (b) two watchmen, (c) three watchmen.}
\label{fig:opt-solutions}
\end{figure*}

For multiple watchman routes, the watchmen share the responsibility of seeing $P$. Thus, we aim to find a ``good'' distribution of responsibilities among the watchmen. 
For two watchmen, we prove that the polygon may be split into two subpolygons such that an optimal solution to the 2-WRP corresponds to an optimal solution to the WRP in each subpolygon.

\section{An Algorithm for Optimal Two Watchmen}
\label{sec:2-wm}
In this section, we investigate the 2-WRP. 
Let us first state some properties of two optimal watchman routes in staircase polygons. 

\begin{lemma}
\label{lemma:properties}
    There exists an optimal solution $(w_1^*, w_2^*)$ to the $2$-WRP in a staircase polygon~$P$ that satisfies the following properties: 
    \begin{enumerate}
        \item $w_1^*$ and $w_2^*$ do not have any common $x$- and $y$-coordinate. 
        \item $w_1^*$ visits the essential cuts $\hbot, \vleft$, and $w_2^*$ visits the essential cuts $\htop, \vright$.
        \item There exists a pair of reflex vertices $(r, r')$ with $r$ on the floor and $r'$ on the ceiling, such that $\overline{r r'}$ separates $w_1^*$ and $w_2^*$, see Figure~\ref{fig:opt-solutions}(b).
    \end{enumerate}
\end{lemma}
\begin{proof}
    First, we prove Properties~1 and~2. 
    Let $(w_1,w_2)$ be an optimal solution to the 2-WRP in $P$. Since $P$ is seen, each of the four essential cuts is visited by some watchman. 
    This can be done in four combinatorially different ways. 
    We distinguish these four cases and prove that in an optimal solution only Case~4 given by Property~2 occurs. 
    If one essential cut is dominated by another one, then both of them are visited by the same watchman (that is, not all four cases are possible). 
    
    \textbf{Case~1:} $w_1$ visits all 4 extensions (Figure~\ref{fig:4-cases}(a)). 
    Then, $w_1$ is a watchman route in $P$, and $w_2$ is redundant. 
    Therefore, assume that $\norm{w_2} \leq \norm{w_1}$, and assume that $\htop$ is non-dominated. 
    Cut $P$ into two subpolygons along $\htop$, and denote the subpolygon below $\htop$ by $P_1$, the one above $\htop$ by $P_2$. Then, $P_1$ is seen by $w_1$. $P_2$ is star-shaped and may therefore be guarded by a single watchman $w_2'$ with route length~0. Hence, replacing $w_1$ with an optimal watchman route $w_1'$ in $P_1$ yields a solution $(w_1', w_2')$ for $P$ that is shorter than $(w_1, w_2)$, both for the min-sum and for the min-max criterion: $w_1'$ will never visit $\htop$, but the essential cuts of $P_1$, and is thus strictly shorter than $w_1$; $0 = \norm{w_2'} \leq \norm{w_2}$, with equality in case $w_2$ was point-shaped. 
    This contradicts the optimality of $(w_1, w_2)$.

    \textbf{Case~2:} $w_1$ visits three essential cuts, $w_2$ visits the fourth one (Figure~\ref{fig:4-cases}(b)). 
    Assume w.l.o.g. that $w_1$ visits $\hbot, \vleft$, and $\htop$. We shorten $w_1$ and $w_2$ as in Case~1. 
    
    \textbf{Case 3:} $w_1$ visits $\vleft$ and $\htop$, and $w_2$ visits $\hbot$ and $\vright$ (Figure~\ref{fig:4-cases}(c)). 
    If the two routes have neither $x$- nor $y$-overlap, then
    observe that $w_2$ lies to the right of the rightmost point of $w_1$. 
    We translate $w_1$ such that its lowest intersection with $\vleft$ lies in the point of intersection of $\hbot$ with $\vleft$. Analogously, we move $w_2$ such that its uppermost intersection with $\vright$ lies in the point of intersection of $\htop$ with $\vright$. 
    By construction, the subpolygon below $\hbot$ (above $\htop$), and the subpolygon left of $\vleft$ (right of $\vright$), are star-shaped and $w_1$ ($w_2$) visits its kernel. 
    Let $P'$ be the subpolygon between the essential cuts. Watchman $w_1$ ($w_2$) sees the ceiling (floor) of $P'$, and moving it vertically does not affect this except possibly losing sight of an interval on $\htop$ ($\hbot$), which will be seen by the translated $w_2$ ($w_1$). Hence, by Lemma 3.1 in~\cite{nilssonpacker2024} (which states that a simple polygon is seen by two watchmen if its boundary is seen), $P'$ is seen as well. 

    Assume now w.l.o.g.\ that there is $y$-overlap between the two routes. 
    Let $q_1$ be the point of intersection of $\vleft$ with $\hbot$, and let $q_2$ be the point of intersection of $\vright$ with $\htop$. 
    We distinguish between the min-sum and the min-max criterion. 
    
    For the min-max case, 
    consider the vertical segment $s_1$ between $q_1$ and the intersection of $\vleft$ with $w_1$, and the vertical segment $s_2$ between $q_2$ and the intersection of $\vright$ with $w_2$, and assume w.l.o.g.\ that $\norm{s_1} \geq \norm{s_2}$. Consider the route $w_2 \cup s_2$, and shorten it by removing the connected part of length $\norm{s_2}$ that starts at $\hbot$ to obtain a new route $w_2'$ of the same length as $w_2$. Then we substitute $w_1$ by a vertical segment $w_1'$ of length $\norm{s_2}$ with lowest point $q_1$. The new route $w_1'$ is shorter than $w_2'$. Moreover, $P$ is seen from $(w_1', w_2')$ because the routes touch $q_1$ and $q_2$, thereby covering $P \setminus P'$, and $P'$ is covered since $w_1'$ and $w_2'$ have $y$-overlap. 

    For the min-sum case, observe that $\norm{s_1} \leq \norm{w_2}$. We construct a new route $w_1' = w_1 \cup s_1$, and replace $w_2$ by a route $w_2'$ of length 0 at its endpoint on $\vright$. Then, the sum of the route lengths does not increase, and the routes satisfy the conditions of Case 2.
    
    \textbf{Case 4:} $w_1$ visits $\hbot$ and $\vleft$, and $w_2$ visits $\htop$ and $\vright$ (Figure~\ref{fig:4-cases}(d)). 
    Assume w.l.o.g.\ that $w_1$ and $w_2$ have some $x$-overlap and let $\ell$ be the leftmost vertical line that intersects both routes. 
    Then, $\ell$ cuts $P$ into subpolygons which both are staircase polygons. Denote the subpolygon to the left of $\ell$ with $P_1$, and the one to the right of $\ell$ with $P_2$. 
    Let $w_1'$ be the watchman route that consists of the part of $w_1$ that lies in $P_1$, together with the straight-line segment between the points of intersection of $w_1$ with $\ell$ (this may also be just a single point). By Observation~\ref{obs:interval}, $P_1$ is seen by $w_1'$. Similarly, $P_2$ is seen by $w_2$. 
    Thus, $(w_1', w_2)$ is a shorter solution for $P$ than $(w_1, w_2)$. 
    Furthermore, these routes may still be improved: If both $P_1$ and $P_2$ are seen, then $P$ is seen. We may thus replace $w_1'$ with an optimal watchman route in $P_1$, and $w_2$ with an optimal watchman route in $P_2$. None of these optimal routes touches $\ell$, and therefore they together yield a solution for $P$ that is shorter than $(w_1, w_2)$ where both routes do not share a common $x$-coordinate. 
    \begin{figure}
    \centering
        \hfill
        \subfigimg[width=0.4\linewidth, page=2]{(a)}{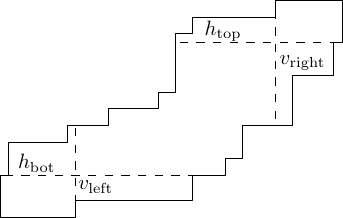}\hfill
        \subfigimg[width=0.4\linewidth, page=5]{(b)}{essential_extensions.pdf}
        \hfill\mbox{}\\ \vspace{0.3cm}\hfill
        \subfigimg[width=0.4\linewidth, page=3]{(c)}{essential_extensions.pdf}\hfill
        \subfigimg[width=0.4\linewidth, page=4]{(d)}{essential_extensions.pdf}
        \hfill\mbox{}
        \caption{Four possibilities for $w_1$ and $w_2$ to visit the up to four essential cuts. }
        \label{fig:4-cases}
    \end{figure}
    
    Let us now show Property 3 to finish the proof. 
    For this, assume that $w_1$ and $w_2$ are two optimal watchman routes that satisfy Properties 1 and 2. Consider all vertical extensions of reflex vertices ordered from left to right. Assume that the interior of every strip between two consecutive such extensions is entered by $w_1$ and $w_2$. Then, we may shorten both $w_1$ and $w_2$ 
    as
    in Case~4 (with the two vertical extensions of the entered strip playing the role of $\ell$), without losing visibility, as both tours will see the complete rectangle spanned by the two extensions, leading to a contradiction. Hence, there exists an empty strip. 
    The diagonal connecting the two opposite reflex vertices of the corresponding extensions does not intersect any of $w_1$ or~$w_2$.
\end{proof}

In the following, we always assume that an optimal solution $(w_1^*, w_2^*)$ obeys Properties~1--3 of Lemma~\ref{lemma:properties}. In particular, $w_1^*$ lies below and to the left of~$w_2^*$.

\begin{lemma} \label{lemma:one-watchman-per-edge}
    In an optimal solution to the $2$-WRP in a staircase polygon, for every polygon edge there exists a watchman that sees the edge completely. 
\end{lemma}
\begin{proof}
    Let $(w_1^*, w_2^*)$ be an optimal solution, and consider $w_1^*$. As soon as it crosses the extension of a horizontal floor edge $e$, it sees $e$ completely since nothing blocks the visibility between $w_1^*$ and $e$ along $e$'s extension. Similarly, $w_1^*$ sees a vertical edge on the ceiling completely as soon as it crosses the edge's extension. Before crossing the extension, $w_1^*$ does not see the respective edge at all. 
    Hence, for any horizontal floor edge (vertical ceiling edge) $e$, if $w_1^*$ sees any point on $e$, then it sees all points of $e$. 
    Similarly, for any horizontal ceiling edge (vertical floor edge) $e$, if $w_2^*$ sees any point on $e$, then it sees all points of $e$. 
    Assume w.l.o.g.\ that there is a horizontal floor edge $e$ such that no point on $e$ is seen by $w_1^*$. Then, $w_2^*$ sees $e$ completely as otherwise there are points on $e$ that are not seen at all. 
\end{proof}
We can therefore separate the two optimal routes.

\begin{lemma}
\label{lem:unique-diagonal}
    Let $(w_1^*, w_2^*)$ be an optimal solution in a staircase polygon~$P$\!.
    There exists a unique diagonal between a vertex on the floor and a vertex on the ceiling that cuts $P$ into two subpolygons $P_1$ and $P_2$ such that $w_1^*$ sees $P_1$, and $w_2^*$ sees~$P_2$. 
\end{lemma}
\begin{proof}
    By Lemma~\ref{lemma:one-watchman-per-edge}, every edge is completely seen by a watchman. For a chain of consecutive edges on the floor or ceiling, there cannot be an alteration in the responsibility of the watchmen: Let $e_i$, $e_{i+1}$, and $e_{i+2}$ be three consecutive edges (on the floor or ceiling). If one watchman sees $e_i$ and $e_{i+2}$ completely, then it also sees $e_{i+1}$. 
    Hence, there exist vertices on the floor and the ceiling such that $w_1^*$ sees all edges that lie below and to the left of them completely, and $w_2^*$ sees all edges that lie above and to the right of them completely. 
    We call such vertices \emph{breaking points} and show that there exist two breaking points, one on the floor and one on the ceiling, that see each other---these define the unique diagonal. 
    Assume that this is not the case. Let $b_f$ be the lowest-leftmost breaking point on the floor, and $b_c$ be the upper-rightmost breaking point on the ceiling. W.l.o.g., assume that all breaking points on the floor lie to the upper-right of the breaking points on the ceiling (in particular, $b_f$ lies to the upper-right of $b_c$). 

    Since $b_f$ and $b_c$ do not see each other, there exist some edges incident to a reflex vertex $r$ that block the visibility. Assume that these edges lie on the ceiling. Then, the horizontal edge incident to $r$ lies above $b_c$ and below $b_f$, and is seen by $w_2^*$ (by definition of $b_c$). Hence, $w_2^*$ sees the vertical floor edge $v$ that is hit by the horizontal extension through $r$ (as described in the proof of Lemma~\ref{lemma:one-watchman-per-edge}), and thereby also the convex vertex on the lower end of $v$, contradicting the choice of~$b_f$ (being the lowest-leftmost breaking point on the~floor). 
\end{proof}

\begin{figure}
\centering\hfill
\subfigimg[width=0.35\linewidth, page=1]{(a)}{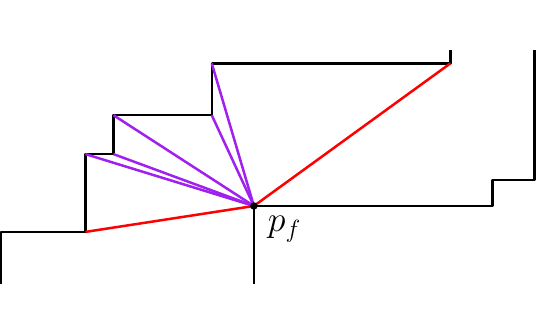}\hfill
\subfigimg[width=0.35\linewidth, page=2]{(b)}{links.pdf}
\hfill\mbox{}
\caption{(a) The candidate diagonals of a reflex vertex $p_f$: two with positive slope (red), and several with negative slope (purple). (b) A diagonal with positive slope (red) that is not a candidate. An optimal route in $P_2$ (gray) needs to visit the same essential cut (dashed) as an optimal watchman route in the subpolygon induced by $\overline{p_cp'_f}$ (purple). }
\label{fig:candidate-diagonals}
\end{figure}

Observe that Lemma~\ref{lemma:one-watchman-per-edge} only holds for two watchmen. For three or more watchmen, some edges may only be seen partially by each watchman in an optimal solution. 
An example is shown in Figure~\ref{fig:opt-solutions}(c). 
The blue watchman is in charge of monitoring a part of a vertical floor edge above the red watchman's visibility region. The yellow watchman does not see this edge at all and would have to walk very far to reach the vertical extension of this edge. 
Therefore, an optimal solution for $m \geq 3$ watchmen may induce a split of the polygon's floor and ceiling into more than $m$ parts each, such that every part is seen by a single watchman. This means that a watchman may be ``in charge of'' more than one contiguous part of the boundary on the floor and ceiling.

We present an algorithm that finds an optimal split, and thus computes an optimal solution for two watchmen in $O(n^2)$ time. 
By 
Lemma~\ref{lem:unique-diagonal}, we consider all diagonals between vertices on the floor and on the ceiling. Any such diagonal splits $P$ into two subpolygons. For each of them, we compute an optimal watchman route using a modified version of the linear-time algorithm 
by Chin and Ntafos~\cite{Optimum_WR}, and then combine the two routes to a solution for the 2-WRP~in~$P$. 

As there are at most quadratically many diagonals to consider, this procedure trivially yields a cubic-time algorithm. However, maintaining a similar structure of the subpolygons by dealing with the diagonals in a certain order allows us to compute many of the watchman routes in amortized constant time. 

To this end, we iterate over the vertices on the floor. For each floor vertex $p_f$, we compute all its diagonals to points on the ceiling, in clockwise order around $p_f$. If $p_f$ is a convex vertex, then all diagonals have a negative slope. If $p_f$ is a reflex vertex, some diagonals have positive slope. 
However, we do not need to consider all diagonals with positive slope, but only those two that are followed or preceded by a 
negative-slope diagonal in clockwise order. We call those, and the diagonals with negative slopes, \emph{candidate diagonals}; see Figure~\ref{fig:candidate-diagonals}(a). 
Every candidate diagonal splits $P$ into two subpolygons, $P_1$ below and $P_2$ above the diagonal.

\begin{lemma} \label{lem:candidatediagonal}
    Any diagonal that is not a candidate diagonal induces a solution that is at least as long as the solution induced by some candidate diagonal. 
\end{lemma}
\begin{proof}
    First, note that a diagonal of positive slope is spanned between two reflex vertices. 
    Consider w.l.o.g.\ a non-candidate diagonal $\overline{p_fp_c}$, as seen in Figure~\ref{fig:candidate-diagonals}(b). 
    Then there is a convex vertex $p'_c$ above $p_c$ that does not yield a diagonal of $p_f$ because $y(p_c) < y(p_f)$. The subpolygon $P_2$ above $\overline{p_fp_c}$ has the horizontal line through $p'_c$ as an essential cut. Hence, the watchman route in $P_2$ has points below this cut. 
    There exists a subpolygon induced by a candidate diagonal (incident to $p_c$ and with the other endpoint $p'_f$ below $p_f$) that also has the horizontal line through $p'_c$ as an essential cut. 
    The watchman route in the subpolygon above 
    $\overline{p_cp'_f}$ remains the same, and the watchman route in the subpolygon below is not longer than the one induced by~$\overline{p_fp_c}$.     
\end{proof}

Now we compute a solution for each candidate diagonal in the following manner.

\textbf{Step~1:} Consider a diagonal with negative slope. 
    Cutting along this diagonal creates only convex vertices in each subpolygon, hence all four essential cuts per subpolygon are rectilinear. 
    The watchman routes touch these extensions, but do not cross them~\cite{Optimum_WR}. 
    We compute the optimal solutions for the subpolygons induced by the first diagonal in clockwise order in linear time by Theorem~\ref{thm:linear-opt}. 
    In addition, we compute two shortest-path-tree data structures~\cite{HarTar:treeancestors}. One is rooted at the first reflex vertex on the floor and stores the shortest paths to all other floor vertices, the other one is rooted at the first reflex vertex on the ceiling and stores the shortest paths to all other ceiling vertices. 
    
    Then, for each diagonal in order, we update the solution in the following way. 
    Moving from one diagonal to the next (i.e., moving from one vertex on the ceiling to the next) alters either the essential cut $\vright\!(P_1)$ of $P_1$, or the essential cut $\hbot\!(P_2)$ of $P_2$. During this movement, any reflex vertex on the ceiling touched by the route can only be released once per vertex $p_f$, and they are released from right to left. 
    Similarly, any reflex vertex on the floor can be added as an anchor point only once per vertex $p_f$, and they get added from left to right. 
    Hence, the number of updates per vertex $p_f$ is at most linear. 

    When updating the route $w_1$ in $P_1$, we move from one vertical extension $\vright\!(P_1)$ to the next one $v'_{\text{right}}\!(P_1)$. 
    We use the shortest-path-tree of the floor to check whether vertices on the floor get added to, and the shortest-path-tree of the ceiling to check whether vertices on the ceiling get released from the route. 
    This can be done in amortized constant time~\cite{HarTar:treeancestors}. 

\textbf{Step 2:} If $p_f$ is a reflex vertex, we need to consider also the two candidate diagonals with positive slope. 
    Here, the subpolygons' essential cuts differ from those of a staircase polygon: There is exactly one non-rectilinear essential cut, namely the extension of the diagonal. We may nevertheless compute an optimal solution, using the algorithm by Chin and Ntafos~\cite{Optimum_WR}.
    This algorithm defines a set of essential cuts, along which the polygon is reflected. Computing the shortest path from one of these essential cuts to its copy yields the shortest watchman route in the original polygon. 
    Since there are at most five essential cuts, we can try all combinations of subsegments of these essential cuts and apply the Chin-and-Ntafos reduction which takes linear time in each of these constant number of cases. 

Thus, the computations for each vertex $p_f$ take amortized linear time. 
As we do this for every vertex on the floor, there are linearly many vertices to consider. With this, we get an optimal solution to the 2-watchman route problem in staircase polygons. 

\begin{theorem}
    An optimal solution to the $2$-WRP in staircase polygons can be computed in $O(n^2)$ time. 
\end{theorem}

\section{An Approximation Algorithm for Min-Max Multiple Watchman Routes}
For the general case of $m$ watchmen, we propose an approximation algorithm for the min-max criterion. 
We consider a canonical set of watchman routes that always span the distance between the floor and the ceiling, thus avoiding the situation depicted in Figure~\ref{fig:opt-solutions}(c). 
Using dynamic programming, we can find a set of min-max canonical routes. We formalize this in Definitions~\ref{def:elbow} and~\ref{def:E-route}. 

\begin{definition}\label{def:elbow}
    Let $p_c$ be a convex vertex on the ceiling, and $p_f$ be a convex vertex on the floor. 
    Let $v_c$ be the maximal vertical line segment in $P$ incident to $p_c$, and let $h_f$ be the maximal horizontal line segment in $P$ incident to $p_f$. If $v_c$ and $h_f$ intersect in a point $s$, we define the \emph{left arm} of $p_c$ and $p_f$ to be the segment pair $\overline{p_cs} \cup \overline{p_fs}$. 
    Then, the \emph{left elbow} of $p_c$ and $p_f$ is the closure of the part of the left arm that lies strictly in the interior of $P$\!, and $s$ is its \emph{(left) elbow point}. 

    Symmetrically, we define the \emph{right arm} and the \emph{right elbow} with its \emph{(right) elbow point} of $p_c$ and $p_f$ by considering the horizontal line segment incident to $p_c$ and the vertical line segment incident to $p_f$. 
\end{definition}

Given two elbows $E$ and $E'$, where $E$ is defined by the two vertices $p_c$ and $p_f$ on the ceiling and the floor, and $E'$ is defined by the two vertices $p_c'$ and $p_f'$ on the ceiling and the floor, respectively, where $p_c'$ does not lie to the left of $p_c$ and $p_f'$ does not lie to the left of $p_f$, we say that $E$ \emph{precedes} $E'$ and denote this by $E \prec E'$. 

For our approximation algorithm, we consider watchman routes that consist of a left elbow $E$, a right elbow $E'$ with $E \prec E'$\!, and the two shortest paths that connect the ceiling endpoints of the two elbows, and the floor endpoints of the two elbows, see Figure~\ref{fig:E-route}.
We call such a route an \emph{E-route}, and denote its length by~$\|T(E,E')\|$. 
\begin{figure}
\centering
    \includegraphics[width=0.41\linewidth, page=1]{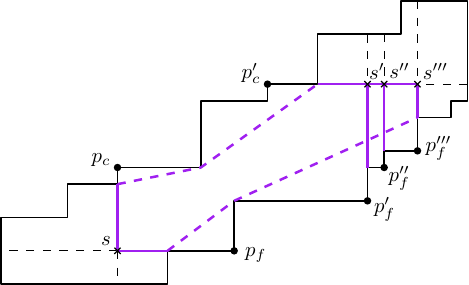}
    \caption{A left elbow defined by $p_c$ and $p_f$, and right elbows defined by $p_c'$ and $p_f'$, $p_f''$, and $p_f'''$ (purple). } 
    \label{fig:E-route}
\end{figure}

\begin{definition}\label{def:E-route}
    A set of watchman routes is called a \emph{canonical solution} if every watchman route is an E-route and every point in $P$ is seen from at least one route. 
\end{definition}

By Observation~\ref{obs:interval}, every point in $P$ that lies between the leftmost and the rightmost point, or the lowest and the uppermost point of an E-route is seen from the route. Thus, we only need to make sure that the points between two E-routes are seen as well. 
Let $E'$ be a right elbow and $E$ be a left elbow, $E' \prec E$. We say that $E'$ and $E$ \emph{partner} each other if every part in $P$ that lies to the right or above $E'$ and to the left or below $E$ is seen from some point on $E'$ or~$E$. 
To find a shortest canonical solution for a polygon $P$\!, we consider all elbows in $P$, sorted by the coordinates of their elbow points, first from left to right and then 
from bottom to top. 
\newcommand{\calL}[1]{\ensuremath{{\cal L}\big(#1\big)}}
\newcommand{\calE}{\ensuremath{{\cal E}}}
\newcommand{\partner}[1]{\ensuremath{L^{\!\rightarrow}\hspace*{-0em}(#1)}}

Let $E_l$ be a left elbow from which every point in $P$ that lies to the left or below $E_l$ is seen, and let $E_r$ be a right elbow from which every point in $P$ that lies to the right or above $E_r$ is seen. We call such an elbow an \emph{outer} elbow. 
Assume $E_l \prec E_r$. 
$E_l$ and $E_r$ define an E-route and, furthermore, every pair of convex vertices on the floor and ceiling that lie between the vertices that define $E_l$ and $E_r$ may define a left and a right elbow. Let \calE\ be the set of elbows between $E_l$ and $E_r$, and let $R(\calE)$ be the set of right elbows in~\calE. 
\partner{\calE,E'} is the set of left elbows in \calE\ partnering the right elbow~$E'$. $A(\calE,E)$ is the set of all elbows in \calE\ preceding the right elbow~$E$.
The following recurrence computes the length of a min-max set of at most $m$~E-routes. 

{
\begin{equation*}
\calL{\calE \!,E_l,E_r,m} =
\nonumber
\hspace*{-2.5ex}
\min_{\stackrel{\mbox{\scriptsize$E_i\in R(\calE)$}}{\mbox{\scriptsize$E_j\in\partner{\calE,E_i}$}}}
\hspace*{-0.5ex}
\left\{
\hspace*{-1.0ex}
\begin{array}{lr}
   \max 
   \Big\{ \big\|T(E_j,E_r)\big\|, 
          &\\\hspace{3.25ex}
          \calL{A(\calE,E_i), E_l, E_i, m-1} 
   \Big\}
& \hspace*{-0.75ex}
\mbox{if $m>1$}
\\
\big\|T(E_l,E_r)\big\|
& \hspace*{-0.75ex}
\mbox{if $m\geq1$}
\end{array}
\right.
\end{equation*}
}
We apply dynamic programming to the recurrence, using all possible outer left elbows $E_l$ and outermost right elbows $E_r$. 
For this, we preprocess the polygon to obtain a data structure that allows us to test whether two elbows are partnering in constant time. 
For each extension of a ceiling edge, we identify the edge on the floor that it hits, 
and similarly for the extensions of floor edges. This takes linear time. 

Let $E'$ be a right elbow of vertices $p'_c$ and $p'_f$ on the ceiling and floor, respectively, let $E$ be a left elbow of vertices $p_c$ and $p_f$ on the ceiling and floor, respectively, and let $E' \prec E$. 
We can determine whether $E'$ and $E$ see all points between them by considering the edges that are hit by the extensions aligned with the elbows. If there is an unseen convex vertex between them, then $E'$ and $E$ are not partnering. Otherwise, the relevant region between them has constant complexity and we can determine visibility in constant time. 

\begin{lemma}\label{lem:e-route-optimality}
   \mbox{Recurrence $\calL{\calE,E_l,E_r,m}$ gives the length} of a min-max canonical solution of up to $m$ E-routes using the outer elbows $E_l$ and $E_r$ in $O(mn^4)$ time. 
\end{lemma}
\begin{proof}
All of $P$ is seen because of Observation~\ref{obs:interval}, the choice of $E_l$ and $E_r$, and the fact that for every right elbow $E_i$ that is selected, the left elbow $E_j$ partners $E_i$. 
The recurrence considers all possible combinations of left and right elbow pairs, and takes the minimum length canonical solution defined by these. 
Since there are up to $O(n^2)$ left and right elbows, there are up to $O(n^4)$ pairs to consider. 
Testing whether a left and a right elbow are partnering each other can be done in constant time using linear time preprocessing. 
Moreover, the recursion splits the polygon into at most $m$ subpolygons, thus the total runtime is $O(mn^4)$. 
\end{proof}

As there are $O(n^2)$ outer left elbows, and $O(n^2)$ outer right elbows, running the algorithm for all such pairs then takes $O(mn^8)$ time and computes a shortest canonical solution. This dominates the cost of preprocessing. 

\begin{lemma}\label{lem:approx-factor}
    Let $W^*$ be an optimal solution to the $m$-WRP in a staircase polygon $P$, and let $W$ be the canonical solution computed by the approximation algorithm. Then, $\lVert W \rVert \leq \lVert W^* \rVert + 4(\hMax + \vMax)$, where $\hMax$ is the length of a longest horizontal line segment in $P$, and $\vMax$ is the length of a longest vertical line segment in~$P$. 
\end{lemma}
\begin{proof}
    Let $w \in W^*$\!\!, and transform $w$ as follows: from the leftmost point $p_c$ of $w$, extend a maximal segment towards the left until it hits a vertical edge of the ceiling at a point $q_c$. Similarly, from its lowest point $p_f$ extend a maximal segment downwards until it hits a horizontal edge on the floor at a point $q_f$. Since $w$ is polygonal, there is a (possibly empty) subpath of $w$ that is both below $p_c$ and to the left of $p_f$. We denote this subpath $w_d\subseteq w$.
    Let $E$ be the left elbow defined by the convex vertices of the edges that contain $q_c$ and $q_f$ and let $r_c$ and $r_f$ denote the reflex vertices of these edges respectively. 
    
    We replace the subpath $w_d$ of $w$ by $\overline{p_c q_c} \cup \overline{q_c r_c} \cup E \cup \overline{r_f q_f} \cup \overline{q_f p_f}$, thereby extending it by the additive term $\leq2(\hMax + \vMax)$.
    Let $p'_c$ and $p'_f$ be the topmost and rightmost points of $w$, respectively. Let $w_u$ be the (possibly empty) subpath to the right of $p'_c$ and above $p'_f$. Let $q'_c$ be the point directly above $p'_c$ and  $q'_f$ be the point to the right of $p'_f$. Define $r'_c$ and $r'_f$ as the reflex vertices of the edges containing $q'_c$ and $q'_f$ and let $E'$ be the corresponding right elbow. We make the analogous replacement of $w_u$ for this case; see Figure~\ref{fig:transformation}, giving us a route~$w'$. 
    The new route $w'$ has length at most $\lVert w \rVert + 4(\hMax + \vMax)$, by the previous~argument.
\end{proof}
    \begin{figure}[h]
        \centering
        \includegraphics[width=0.35\textwidth]{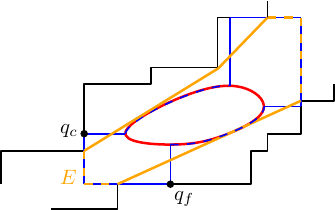}
        \caption{Transforming an optimal route $w$ (red) into an E-route (orange). }
        \label{fig:transformation}
    \end{figure} 
The algorithm, albeit having approximation factor one, since the approximation error is additive, will never compute a route of length zero. Hence, it cannot be used to obtain a polynomial time solution to minimum point guarding in staircase polygons. A polynomial time solution for this problem remains elusive.


{\small
\bibliographystyle{abbrv}

\bibliography{two-WRP-staircase}
}

\end{document}